\documentclass[runningheads,12pt]{llncs}
\usepackage[pdftex]{graphicx}
\usepackage{amsmath,amssymb,amsfonts}
\usepackage[textsize=small,textwidth=1.3in]{todonotes}
\usepackage[english]{babel}

\usepackage{hyperref}

\usepackage{authblk}

\usepackage{multirow}

\usepackage{tikz}
\usepackage{color, colortbl}
\usepackage{multirow}
\usepackage{easybmat}

\usepackage{amsthm}

\newcommand{\fq}{{\mathbb F}_q}

\newcommand{\bc}{{\bf c}}

\newcommand{\bv}{{\bf v}}
\newcommand{\bx}{{\bf x}}

\newcommand{\rk}{\rm{rank}}

\newcommand{\supp}{\rm{supp}}

\begin{document}
\title{Entanglement-assisted Quantum Codes from Algebraic Geometry Codes}

\titlerunning{QUENTA codes from AG Codes}
%
\author{Francisco Revson F. Pereira\inst{1,2} \and
Ruud Pellikaan\inst{1} \and Giuliano Gadioli La Guardia\inst{3} \and Francisco Marcos de Assis\inst{2}}
\authorrunning{F. R. Fernandes Pereira, \emph{et al.}}

\institute{Department of Mathematics and Computing Science, Eindhoven University of Technology, Eindhoven, The Netherlands. \and
	   Department of Electrical Engineering,
	   Federal University of Campina Grande, Campina Grande, Para\'iba, Brazil. \and
Department of Mathematics and Statistics, State University of Ponta Grossa, Ponta Grossa, Paran\'a, Brazil.\\
\email{revson.ee@gmail.com}}
%
\maketitle              
\begin{abstract}

Quantum error correcting codes play the role of suppressing noise and decoherence
in quantum systems by introducing redundancy. Some strategies can be used to improve
the parameters of these codes. For example, entanglement can provide a way for quantum
error correcting codes to achieve higher rates than the one obtained via the traditional
stabilizer formalism. Such codes are called entanglement-assisted quantum (QUENTA) codes.
In this paper, we use algebraic geometry codes to construct
several families of QUENTA codes via the Euclidean and the Hermitian construction.
Two of the families created have maximal entanglement and have quantum Singleton defect equal to
zero or one.
Comparing the other families with the codes with the respective quantum Gilbert-Varshamov bound,
we show that our codes have a rate that surpasses that bound.
At the end, asymptotically good towers of linear complementary dual codes are used
to obtain asymptotically good families of maximal
entanglement QUENTA codes. Furthermore, a simple comparison with the quantum Gilbert-Varshamov bound
demonstrates that using our construction it is possible to create an asymptotically family
of QUENTA codes that exceeds this bound.\footnote{This paper was presented in part at the 2019 IEEE International Symposium on Information Theory.}

\keywords{Quantum Codes \and Algebraic Geometry Codes \and Maximal Distance Separable
\and Maximal Entanglement \and Asymptotically Good.}
\end{abstract}
%
%

\section{Introduction}
\label{sec:Introduction}
\noindent

It is generally accepted that the prospect of practical large-scale quantum computers and the use of
quantum communication are only possible with the implementation of quantum error correcting codes.
Suppressing noise and decoherence can be done via Quantum error correcting codes.
The capability of correcting errors of such codes can be improved if
it is possible to have pre-shared entanglement states. They are known as Entanglement-Assisted
Quantum (QUENTA) codes, also denoted by EAQECC's in the literature, but we prefer the acronym QUENTA since 
that is more easy to pronounce. Additionally, this class of codes achieves the hashing bound \cite{Wilde:2014,Lai:2014}
and violates the quantum Hamming bound \cite{Li:2014}. The first QUENTA codes were proposed by Bowen \cite{Bowen:2002}
followed by the work from Fattal, \emph{et al.} \cite{Fattal:2004}. The stabilizer formalism of QUENTA codes
was created by Brun \emph{et al.} \cite{Brun:2006}, where they showed that QUENTA codes paradigm does not
require the dual-containing constraint as the standard quantum error-correcting code does \cite{LidarBrun:Book}.
Wilde and Brun \cite{Wilde:2008} proposed two methods to create QUENTA codes from classical codes, which
are named in this paper the Euclidean construction method and the Hermitian construction method. These
methods were recently generalized by Galindo, \emph{et al.} \cite{Galindo:2019}

After these works of Brun \emph{et al.}, many articles have focused on the construction of QUENTA
codes based on classical linear codes \cite{Wilde:2008,Chen:2017,Lu:2017,Guenda:2018,Liu:2019}.
However, the analysis of $q$-ary QUENTA codes was taken into account only recently
\cite{Fan:2016,Chen:2017,Lu:2018,Chen:2018,Guenda:2018,Liu_ArXiv:2018,Guenda_ArXiv:2018,Liu:2019}.
The majority of them utilized constacyclic codes \cite{Fan:2016,Chen:2018,Lu:2018} or
negacyclic codes \cite{Chen:2017,Lu:2018} as the classical counterpart. Since the length of the
classical codes is normally proportional to the square of the size of the field, most of the quantum codes from the
previous works have a length that is proportional to the square of the size of the finite field.
Hence, there is no result in the literature with QUENTA codes having length proportional to a greater power of
the cardinality of the finite field. In addition, it has not been shown previously that there exists a family of
asymptotically good maximal entanglement QUENTA codes attaining quantum Gilbert-Varshamov bound \cite{Galindo:2019}.
Such a family can be used to achieve the hashing bound. A
possible approach to solve both questions is using algebraic geometry (AG) codes as the classical
counterpart to construct QUENTA codes.

The AG codes were invented by Goppa~\cite{Goppa:1981}. An important property
of these codes is that its parameters can be calculated via the degree of a divisor, which allows a
direct description of the code. The first result of this paper comes from these properties. We show two methods to create new AG codes from
old ones via intersection and union of divisors. As will be shown, the former ``new codes from old''
construction is crucial when two AG codes are used to derive QUENTA codes. To derive the QUENTA codes in this
paper, it is necessary to define some mathematical tools and the relation between them and the parameters of
QUENTA codes.

First, we introduce the idea of intersection and union of divisor and how this concepts
can be used to construct new AG codes from old ones. In addition,
it is shown that the amount of entanglement in the Euclidean construction
method of QUENTA codes can be described via the intersection of two classical codes used.
The practicality of such description is presented by applying our method to
AG codes derived from three curves: the projective line (rational function
field), the Hermitian Curve, and the elliptic curve. The QUENTA codes derived from the first (third) curve are shown to be
maximal distance separable (MDS) codes (almost MDS), i.e., the minimal distance of these codes achieve the quantum Singleton
bound (differs from the quantum Singleton bound by at most one unit), and maximal entanglement. These codes can be employed to achieve entanglement-assisted
quantum capacity of a depolarizing channel \cite{Bowen:2002,Devetak:2008,Lu:2015}.
For the Hermitian curve, a comparative analysis with the codes in the literature shows that our codes have better parameters.

The use of AG codes in the Hermitian construction method for QUENTA codes does not follows the same procedure
of the Euclidean one. The reason for this is that there is no general characterization of the Hermitian dual
code of an AG code. Thus, to determine the parameters of QUENTA codes derived from the AG codes used, 
the intersection of the bases of two AG codes is computed. The curve used to construct the QUENTA codes is the
projective line. The QUENTA codes created are MDS codes and also have maximal
entanglement.

Lastly, asymptotically good families of LCD codes are used to construct asymptotically good families of
QUENTA codes that have maximal entanglement. Using AG codes from a tower of function fields
that attain the Drinfeld-Vladut bound \cite{Stichtenoth:2009} we show that the QUENTA codes in this paper
surpass the quantum Gilbert-Varshamov bound \cite{Galindo:2019}.


The paper is organized as follows. In Section~\ref{sec:Preliminaries}, we describe what needs to be known
about AG codes, so that they can be applied to the generalization of the construction methods of QUENTA codes
from Wilde and Brun \cite{Wilde:2008} proposed by Galindo, \emph{et al.} \cite{Galindo:2019}. 
In this section, two methods to construct new AG codes from old ones are shown.
Afterwards, several new families of QUENTA codes are derived from AG codes. These derivations come from the
Euclidean and the Hermitian construction methods for QUENTA codes with the use of three
different types of curves. In Section~\ref{sec:codComp}, we compare the codes in this paper with the quantum
Singleton bound and with other quantum codes in the literature. In particular, it is shown that three families
of QUENTA codes constructed are MDS or almost MDS. In Section~\ref{sec:AsymptoticallyGood},
we show that there exists families of QUENTA codes that surpass the quantum Gilbert-Varshamov bound.
Lastly, the conclusion is given in Section~\ref{sec:Conclusion}.

\emph{Notation.} Throughout this paper, $p$ denotes a prime number and $q$ is a power of $p$.
$F/\mathbb{F}_q$ denotes an algebraic function field over $\mathbb{F}_q$ of genus $g$, where $\mathbb{F}_q$
denotes the finite field with $q$ elements. A linear code $C$ with parameters $[n,k,d]_q$ is a $k$-dimensional subspace of
$\mathbb{F}_q^n$ with minimum distance $d$. Lastly, an $[[n,k,d;c]]_q$ quantum code is a
$q^k$-dimensional subspace of $\mathbb{C}^{q^n}$ with minimum distance $d$ that utilizes $c$ pre-shared
entanglement pairs.


\section{Preliminaries}
\label{sec:Preliminaries}
\noindent
In this section, we introduce some ideas related to linear complementary dual (LCD) codes,
algebraic geometry (AG) codes and entanglement-assisted quantum (QUENTA) codes.
Before we give a description of LCD codes, a definition for
the Euclidean and the Hermitian dual of a code needs to be given.

\begin{definition}
%
    Let $C$ be a linear code over $\mathbb{F}_q$ with length $n$. The (Euclidean) dual of $C$ is defined as

    \begin{equation}
    C^\perp = \{ \ \bx \in \mathbb{F}_q^n \ | \ \bx\cdot \bc =0 \mbox{ for all } \bc \in C \}.
    \end{equation}
    If the code $C$ is $\mathbb{F}_{q^2}$-linear, then we can define the Hermitian dual of $C$.
    This dual code is defined by

    \begin{equation}
	C^{\perp_h} = \{ \ \bx \in \mathbb{F}_{q^2}^n \ | \ \bx\cdot \bc^q =0 \mbox{ for all } \bc \in C \ \},
    \end{equation}where $\bc^q = (c_1^q, \ldots, c_n^q)$ for $\bc \in\mathbb{F}_{q^2}^n$.

%
\end{definition}

When the intersection between a code and its dual gives only the vector $\bf{0}$, the code is called LCD.
A formal description can be seen below.

\begin{definition}
    The hull of a linear code $C$ is given by $hull (C) = C^\perp \cap C$. The code is called \rm{linear complementary
    dual} (LCD) code if the hull is trivial; i.e, $hull(C) = \{\bf{0}\}$. Similarly,
    $hull_H (C) = C^{\perp_h} \cap C$ and $C$ is called hermitian LCD code if $hull_H (C) = \{\bf{0}\}$.
\end{definition}

The class of LCD codes is a possible way to construct QUENTA codes that have maximal entanglement and
asymptotically good families of QUENTA codes(see Sections~\ref{sec:NewConstructions} and \ref{sec:AsymptoticallyGood}).

%
%

\subsection{Algebraic-Geometry codes}

Let $F/\mathbb{F}_q$ be an algebraic function field of genus $g$. A
place $P$ of $F/\mathbb{F}_q$ is the maximal ideal of some valuation
ring $\mathcal{O}_P$ of $F/\mathbb{F}_q$. We also define the set of
all places by $\mathbb{P}_F
= \{P | P\textit{ is a place of }F/\mathbb{F}_q\}$.

A divisor of $F/\mathbb{F}_q$ is a formal sum of places given by $D = \sum_{P\in
\mathbb{P}_F} n_P P$, with $n_P\in\mathbb{Z}$, where
almost all $n_P=0$. The support and degree of $D$ are defined as
${\supp}(D)=\{P\in \mathbb{P}_F|n_p\neq 0\}$ and
$\operatorname{deg}(D)= \sum_{P\in \mathbb{P}_F} n_P
\operatorname{deg}(P)$, respectively, where $\operatorname{deg}(P)$
is the degree of the place $P$. When a place has degree one, it is called a rational place.

The discrete valuation corresponding
to a place $P$ is written as $\nu_P$. For every element $f$ of
$F/\mathbb{F}_q$, we can define a principal divisor of $f$ by $(f)
= \sum_{P\in \mathbb{P}_F}\nu_P (f) P$. For $f \in \mathcal{O}_P$, we define
$f(P) \in \mathcal{O}_P/P$ to be the residue class of $f$ modulo $P$;
for $f \in F\setminus \mathcal{O}_P$, we put $f(P) = \infty$.
For a given divisor $G$, we denote the
Riemann-Roch space associated to $G$ by $\mathcal{L}(G) = \{f \in
F^* | (f) \geq -G\}\cup \{0\}$. 
%
%

The given description of Riemann-Roch spaces shows that when we are talking about such spaces we deal with
functions that obey a set of rules which are described by the defining divisor. One natural question that could
arise is the relation between the intersection of two Riemann-Roch spaces and the respective divisor that defines
such a space. Such a result was shown by Munuera and Pellikaan \cite{Munuera:1993}. Before showing it, we need to define
the intersection and union of two divisors, which is done in the following.

\begin{definition}
  Let $G$ and $H$ be divisors over $F/\fq$. If $G = \sum_{P\in\mathbb{P}_F} \nu_P(G) P$ and
  $H = \sum_{P\in\mathbb{P}_F} \nu_P(H) P$, where $P\in\mathbb{P}_F$ is a place,
  then the intersection $G\cap H$ of $G$ and $H$ over $F/\mathbb{F}_q$ is defined as follows

  \begin{equation}
    G\cap H = \sum_{P\in\mathbb{P}_F} \min\{\nu_P(G),\nu_P(H)\}P.
  \end{equation}In addition, the union is given by
  \begin{equation}
    G\cup H = \sum_{P\in\mathbb{P}_F} \max\{\nu_P(G),\nu_P(H)\}P.
  \end{equation}
\end{definition}

\begin{proposition}\cite[Lemma 2.6]{Munuera:1993}
  Let $G$ and $H$ be divisors over $F/\mathbb{F}_q$. Then $\mathcal{L}(G) \cap \mathcal{L}(H) = \mathcal{L}(G\cap H)$.
  \label{prop:MunueraPellikaan}
\end{proposition}

In Section~\ref{sec:NewConstructions} it will be shown that when AG codes are used to construct QUENTA codes,
the amount of entanglement used is related to the dimension of the intersection of the two Riemann-Roch spaces.

For the exact value of the dimension of a Riemann-Roch space and the construction of the dual code
of an AG code, it is necessary to introduce the ideas of differential spaces and canonical divisors.
Let $\Omega_F = \{\omega|\omega \textit{ is a Weil differential of
}F/\mathbb{F}_q\}$ be the differential space of $F/\mathbb{F}_q$. Given a
nonzero differential $\omega$, we denote by $(\omega)=\sum_{P\in \mathbb{P}_F} \nu_P(\omega) P$
the canonical divisor of $\omega$. All canonical divisors are equivalent and have
degree equal to $2g-2$. Furthermore, for a divisor $G$ we define
$\Omega_F(G) = \{\omega \in \Omega_F| \omega = 0 \textit{ or }
(\omega)\geq G\}$, and its dimension as an $\mathbb{F}_q$-vector
space is denoted by $i(G)$.

The dimension of a Riemann-Roch space can be calculated through its defining divisor, the divisor of a Weil
differential and the genus of a curve.

\begin{proposition}\cite[Theorem 1.5.15]{Stichtenoth:2009}(Riemann-Roch Theorem)
Let $W$ be a canonical divisor of $F/\mathbb{F}_q$. Then for each divisor $G$,
the dimension of $\mathcal{L}(G)$ is given by $\ell(G) =
\operatorname{deg}(G) + 1 - g + \ell(W - G),$ where $\operatorname{deg}(G)$ is the
degree of the divisor $G$.
\end{proposition}

Now we define the first AG code utilized in this paper, see Definition~\ref{AgCodes:definitionAG}, and its
parameters, see Proposition~\ref{AgCodes:proposition1}. The definition of such AG codes is given as the image
of a linear map called the evaluation map. The parameters of the AG codes are related to the degrees of divisors,
genus and number of rational places. Thus, with simple arithmetic we can create families of codes,
even when the algebraic function field is fixed.

\begin{definition}
Let $P_1, \ldots, P_n$ be pairwise distinct rational places of
$F/\mathbb{F}_q$ and $D = P_1 + \cdots + P_n$. Choose
a divisor $G$ of $F/\mathbb{F}_q$ such that ${\supp} (G)\cap {\supp} (D) =
\varnothing$. The algebraic-geometry (AG) code $C_{\mathcal{L}}(D, G)$ associated
with the divisors $D$ and $G$ is defined as the image of the linear map
$ev_D\colon \mathcal{L}(G)\rightarrow \mathbb{F}_q^n$ called the evaluation map, where
$ev_D(f) = (f(P_1), \ldots, f(P_n))$; i.e.,
$C_{\mathcal{L}}(D, G)
= \{(f(P_1), \ldots, f(P_n))| f\in \mathcal{L}(G)\}$.
\label{AgCodes:definitionAG}
\end{definition}

\begin{proposition}\cite[Corollary 2.2.3]{Stichtenoth:2009}
Let $F/\mathbb{F}_q$ be a function
field of genus $g$. Then the AG code $C_{\mathcal{L}}(D, G)$ is an
$[n, k, d]$-linear code over $\mathbb{F}_q$ with parameters $k =
\ell(G) - \ell(G-D)\textit{ and } d\geq n- \operatorname{deg} (G)$.
If $2g - 2 < \operatorname{deg}(G) < n$, then $k =
\operatorname{deg}(G) - g + 1$. 
\label{AgCodes:proposition1}
\end{proposition}

The next two propositions present a way to construct new AG codes from old AG codes via the intersection
and union of divisors.
Proposition~\ref{lema_intersection_AG} will be used in Section~\ref{sec:NewConstructions} in order to use AG
codes to create QUENTA codes.

\begin{proposition}
      Let $F/\mathbb{F}_q$ be a function field of genus $g$ and let $D$ be a divisor as in
      Definition~\ref{AgCodes:definitionAG}. If $G_1$ and $G_2$ are two divisors such that
      ${\supp} G_1\cap {\supp} D = \varnothing$, resp. ${\supp} G_2\cap {\supp} D = \varnothing$,
      and $\operatorname{deg}(G_1\cup G_2) < n$, then
      $C_{\mathcal{L}}(D, G_1)\cap C_{\mathcal{L}}(D, G_2) = C_{\mathcal{L}}(D, G_1\cap G_2)$.
      \label{lema_intersection_AG}
\end{proposition}

\begin{proof}
   First of all, consider that ${\bf c}\in C_{\mathcal{L}}(D, G_1)\cap C_{\mathcal{L}}(D, G_2)$, then there exist
   $g_1\in\mathcal{L}(G_1)$ and $g_2\in\mathcal{L}(G_2)$
   such that ${\bf c} = ev_D(g_1) = ev_D(g_2)$, which implies $ev_D(g_1-g_2) = 0$. Since that $g_1 - g_2 \in \mathcal{L}(G_1\cup G_2)$ and
   $\operatorname{deg}(G_1\cup G_2) < n$, then $g_1 = g_2\in\mathcal{L}(G_1\cap G_2)$ by Proposition~\ref{prop:MunueraPellikaan}.
   Consequently, ${\bf c}\in C_{\mathcal{L}}(D, G_1\cap G_2)$.
   The other inclusion is straightforward consequence of Proposition~\ref{prop:MunueraPellikaan}.
\end{proof}

\begin{proposition}
      Let $F/\mathbb{F}_q$ be a function field of genus $g$ and let $D$ be a divisor as in
      Definition~\ref{AgCodes:definitionAG}. If $G_1$ and $G_2$ are two divisors such that
      ${\supp} G_1\cap {\supp} D = \varnothing$, resp. ${\supp} G_2\cap {\supp} D = \varnothing$,
      and $\operatorname{deg}(G_1\cap G_2) > 2g-2$ and $\operatorname{deg}(G_1\cup G_2) < n$, then
      $C_{\mathcal{L}}(D, G_1) + C_{\mathcal{L}}(D, G_2) = C_{\mathcal{L}}(D, G_1\cup G_2)$.
      \label{proposition_union_AG}
\end{proposition}

\begin{proof}
      Lets begin considering the inclusion
      $C_{\mathcal{L}}(D, G_1) + C_{\mathcal{L}}(D, G_2) \subseteq C_{\mathcal{L}}(D, G_1\cup G_2)$. Since
      $G_i\leq G_1\cup G_2$, for $i = 1,2$, then $C_{\mathcal{L}}(D, G_i)\subseteq C_{\mathcal{L}}(D, G_1\cup G_2)$,
      for $i=1,2$. Hence $C_{\mathcal{L}}(D, G_1)+C_{\mathcal{L}}(D, G_1)\subseteq C_{\mathcal{L}}(D, G_1\cup G_2)$.
      On the other hand, notice that $\ell(G_1) + \ell(G_2) = \ell(G_1\cap G_2) + \ell(G_1\cup G_2)$, since that
      $\operatorname{deg}(G_1\cap G_2) > 2g-2$. This implies that
      $\mathcal{L}(G_1\cup G_2) = \mathcal{L}(G_1) + \mathcal{L}(G_2)$
      by Proposition~\ref{prop:MunueraPellikaan}. Now, the proof of the remaining inclusion follows from the
      hypothesis that $\operatorname{deg}(G_1\cup G_2) < n$ and Proposition~\ref{lema_intersection_AG}.
\end{proof}

Another important type of AG code is given in the following.

\begin{definition}
Let $F/\mathbb{F}_q$ be a function field of genus $g$ and let $G$ and $D$ be divisors as in
Definition~\ref{AgCodes:definitionAG}. Then we define the code
$C_\Omega (D, G)$ as $C_\Omega (D, G) =
\{(\textit{res}_{P_1}(\omega), \ldots,$ $\textit{res}_{P_n}(\omega)|
\omega\in \Omega_F(G - D)\}$, where $\textit{res}_{P_i}(\omega)$
denotes the residue of $\omega$ at $P_i$.
\label{AgCodes:definitionOmega}
\end{definition}

\begin{proposition}\cite[Theorem 2.2.7]{Stichtenoth:2009}
Let $C_\Omega (D, G)$ be the AG code from Definition~\ref{AgCodes:definitionOmega}. If $2g - 2 <
\operatorname{deg}(G) < n$, then $C_\Omega (D, G)$ is an $[n, k',
d']$-linear code over $\mathbb{F}_q$, where $k' =
n+g-1-\operatorname{deg}(G)$ and $d' \geq \operatorname{deg}(G) -
(2g -2)$.
\label{AgCodes:proposition2}
\end{proposition}

The relationship between the codes $C_{\mathcal{L}}(D,G)$ and $C_\Omega (D,G)$ is
given in the next proposition.

\begin{proposition}\cite[Proposition 2.2.10]{Stichtenoth:2009}
Let $C_{\mathcal{L}} (D,G)$ be the AG code described in Definition~\ref{AgCodes:definitionAG}. Then
$C_{\Omega} (D, G)$ is its Euclidean dual, i.e., $C_{\mathcal{L}} (D,G)^\bot = C_{\Omega} (D, G)$.
Additionally, if we have a Weil differential $\eta$ such that $\nu_{P_i}(\eta) = -1$ and $\eta_{P_i}(1) = 1$
for all $i=1, \ldots, n$, then $C_{\Omega} (D, G) = C_{\mathcal{L}} (D, G^\perp)$, where $G^\perp = D - G + (\eta)$.
\label{AgCodes:DualCodeWeilDiff}
\end{proposition}


\subsection{Entanglement-assisted quantum codes}

\begin{definition}
  A quantum code $\mathcal{Q}$ is called an $[[n,k,d;c]]_q$ entanglement-assisted quantum (QUENTA) code
  if it encodes $k$ logical qudits into $n$ physical qudits using $c$ copies of maximally entangled states
  and can correct $\lfloor(d-1)/2\rfloor$ quantum errors. The rate of a QUENTA code is given by $k/n$, relative
  distance by $d/n$, and entanglement-assisted rate by $c/n$. Lastly, a QUENTA code is said to have maximal
  entanglement when $c = n-k$.
\end{definition}

Formulating a stabilizer paradigm for QUENTA codes gives a way to use classical codes to construct this
quantum codes \cite{Brun:2014}. In particular, we have the next two procedures by
Galindo, \emph{et al.} \cite{Galindo:2019}.

\begin{proposition}\cite[Theorem 4]{Galindo:2019}
Let $C_1$ and $C_2$ be two linear codes over $\mathbb{F}_q$ with parameters $[n,k_1,d_1]_q$ and
$[n,k_2,d_2]_q$ and parity check matrices $H_1$ and $H_2$, respectively. Then there is a QUENTA code with parameters
$[[n,k_1+k_2-n+c, d; c]]_q$, where $d = \min\{d_H(C_1\setminus(C_1\cap C_2^\perp)), d_H(C_2\setminus(C_1^\perp\cap C_2))\}$,
where $d_H$ is the minimum Hamming weight of the vectors in the set, and

\begin{equation}
  c = {\rk} (H_1 H_2^T) = \dim C_1^\perp - \dim (C_1^\perp\cap C_2)
\end{equation}is the number of required maximally entangled states.
\label{Prep:WildeEuclid}
\end{proposition}

A straightforward application of LCD codes to the Proposition~\ref{Prep:WildeEuclid} can produce some interesting
quantum codes. See Theorem~\ref{preliminaries:theorem1} and Corollary~\ref{corollary2}.

\begin{theorem}
    Let $C_1$ and $C_2$ be two linear codes with parameters $[n,k_1,d_1]_q$ and $[n,k_2,d_2]_q$, respectively, with
    $C_1^\perp\cap C_2 = \{\bf{0}\}$. Then there exists a QUENTA code with parameters $[[n,k_2, \min\{d_1,d_2\};n-k_1]]_q$.
    \label{preliminaries:theorem1}
\end{theorem}

\begin{proof}
            Since that $C_1^\perp\cap C_2 = \{\bf{0}\}$, from Proposition~\ref{Prep:WildeEuclid} we have that
            the QUENTA code constructed from $C_1$ and $C_2$ has parameters $[[n,k_2, \min\{d_1,d_2\};n-k_1]]_q$.
\end{proof}

\begin{corollary}
    Let $C$ be a LCD code with parameters $[n,k,d]_q$. Then there exists a maximal entanglement QUENTA code
    with parameters $[[n,k,d;n-k]]_q$. In particular, if $C$ is MDS then the QUENTA codes is also MDS.
    \label{corollary2}
\end{corollary}

\begin{proof}
            Let $C_1 = C_2 = C$. Since $C$ is LCD, then $\dim(hull(C))=0$. Then, from Theorem~\ref{preliminaries:theorem1},
            we have that there exists a QUENTA code with parameters $[[n,k, d; n-k]]_q$.
\end{proof}


\begin{proposition}\cite[Proposition 3 and Corollary 1]{Galindo:2019}
Let $C$ be a linear codes over $\mathbb{F}_{q^2}$ with parameters $[n,k,d]_{q^2}$,
$H$ be a parity check matrix of $C$, and $H^*$ be the $q$-th power of the transpose matrix of $H$.
Then there is a QUENTA code with parameters
$[[n,2k-n+c, d'; c]]_q$, where $d' = d_H(C\setminus(C\cap C^{\perp_h}))$,
where $d_H$ is the minimum Hamming weight of the vectors in the set, and

\begin{equation}
  c = {\rk} (H H^*) = \dim C^{\perp_h} - \dim (C^{\perp_h}\cap C)
\end{equation}is the number of required maximally entangled states.
\label{Prep:WildeHerm}
\end{proposition}

In the same way as before, it possible to use hermitian LCD codes to derive
QUENTA codes with interesting properties. See the following theorem.

\begin{theorem}
    Let $C$ be a hermitian LCD code with parameters $[n,k,d]_{q^2}$. Then there exists a maximal entanglement QUENTA code
    with parameters $[[n,k,d;n-k]]_q$. In particular, if $C$ is MDS, then the QUENTA code is also MDS.
    \label{preliminaries:proposition2}
\end{theorem}

\begin{proof}
            Since $C$ is a hermitian LCD code, then $\dim(hull_H(C)) = 0$. Therefore, using
            $C$ in the Proposition~\ref{Prep:WildeHerm},
            we have that there exists an QUENTA code with parameters $[[n,k, d; n-k]]_q$.
\end{proof}

A measurement of goodness for a QUENTA code is the quantum Singleton bound (QSB).
Let $[[n,k,d;c]]_q$ be an QUENTA code, then the QSB is given by

\begin{equation}\label{QSB}
            d \leq  \Big{\lfloor}\frac{n-k+c}{2}\Big{\rfloor} + 1.
\end{equation}The difference between the QSB and $d$ is called quantum Singleton defect, which is
$\lfloor\frac{n-k+c}{2}\rfloor + 1 - d$. When the quantum Singleton
defect is equal to zero (resp. one) the code is called a maximum distance separable quantum code (resp. almost maximum
distance separable quantum code) and it is denoted by MDS quantum code (resp. almost MDS quantum code).

%
%

\section{New Construction Methods for QUENTA Codes}
\label{sec:NewConstructions}
\subsection{Euclidean Construction}


In Proposition~\ref{Prep:WildeEuclid}, the connection between the
entanglement in an QUENTA code and the relative hull of two classical codes is shown. However,
the computation of such a hull can be difficult in some cases, however, as we are going to show
in Theorem~\ref{newContructions:dimhull}, this is not the case for AG codes.

The rank of a matrix that is the product of the two parity check matrices
of the classical codes is utilized to construct such a quantum code. However, such rank can be difficult to calculate in
some cases. As it will be shown,
it is possible to, instead of calculating such rank, relate the entanglement with the relative hull between
the two classical codes. For that, we need first to present the connection between the rank in
Proposition~\ref{Prep:WildeEuclid} and the relative hull.

\begin{theorem}
    Let $P_1, \ldots, P_n$ be pairwise distinct rational places of
    $F/\mathbb{F}_q$ and $D = P_1 + \cdots + P_n$. Choose
    divisors $G_1, G_2$ of $F/\mathbb{F}_q$ such that ${\supp} (G_1)\cap {\supp} (D) =
    \varnothing$ and ${\supp} (G_1)\cap {\supp} (D) = \varnothing$. Let
    $C_1 = C_{\mathcal{L}}(D,G_1)$ and $C_2 = C_{\mathcal{L}}(D,G_2)$.
    If $\operatorname{deg}(G_1^\perp \cup G_2)<n$,
    then $\dim(C_1^\perp\cap C_2) = \ell(G_1^\perp\cap G_2)$.
\label{newContructions:dimhull}
\end{theorem}

\begin{proof}
  Since $\operatorname{deg}(G_1^\perp\cup G_2)<n$, we can use
  Proposition~\ref{lema_intersection_AG} for the codes $C_{\mathcal{L}}(D,G_1)^\perp$
  and $C_{\mathcal{L}}(D,G_2)$. Hence, it is easy to see from this proposition that
  $C_{\mathcal{L}}(D,G_1)^\perp\cap C_{\mathcal{L}}(D,G_2) = C_{\mathcal{L}}(D,G_1^\perp\cap G_2)$,
  which implies that
  $\dim(C_{\mathcal{L}}(D,G_1)^\perp\cap C_{\mathcal{L}}(D,G_2)) = \ell(G_1^\perp\cap G_2)$.
\end{proof}

Theorem~\ref{newContructions:dimhull} allows us to use AG codes from any function field to construct QUENTA codes,
which is given in detail in Theorem~\ref{newContructions:AGQUENTA}.
In particular, as it will be shown, we can use AG codes to derive MDS quantum codes and asymptotically good QUENTA codes.

\begin{theorem}
    Let $P_1, \ldots, P_n$ be pairwise distinct rational places of
    $F/\mathbb{F}_q$ and $D = P_1 + \cdots + P_n$. Choose
    divisors $G_1, G_2$ of $F/\mathbb{F}_q$ such that ${\supp} (G_i)\cap supp (D) =
    \varnothing$ and $2g-2<\operatorname{deg}(G_i)<n$, for $i=1,2$. Let
    $C_1 = C_{\mathcal{L}}(D,G_1)$ and $C_2 = C_{\mathcal{L}}(D,G_2)$.
    If $\operatorname{deg}(G_1^\perp\cup G_2)<n$,
    then there exists a QUENTA code with parameters $[[n,\operatorname{deg}(G_1+G_2)-2g+2-n+c, d;c]]_q$, where
    $d\geq n-\max\{\operatorname{deg}(G_1), \operatorname{deg}(G_2)\}$ and
    $c = n+g-1-\operatorname{deg}(G_1)-\ell(G_1^\perp\cap G_2)$.
\label{newContructions:AGQUENTA}
\end{theorem}

\begin{proof}
  First of all, notice that the parameters of the AG codes $C_{\mathcal{L}}(D,G_1)$ and $C_{\mathcal{L}}(D,G_2)$ are
  $[n,\operatorname{deg}(G_1)-g+1, d_1\geq n-\operatorname{deg}(G_1)]_q$ and
  $[n,\operatorname{deg}(G_2)-g+1, d_2\geq n-\operatorname{deg}(G_2)]_q$, respectively, and the dimension of the
  Euclidean dual of $C_{\mathcal{L}}(D,G_1)$ is $n+g-1-\operatorname{deg}(G_1)$, by
  Proposition~\ref{AgCodes:DualCodeWeilDiff}.
  From Theorem~\ref{newContructions:dimhull} we have that
  $\dim(C_1^\perp\cap C_2)) = \ell(G_1^\perp\cap G_2)$. Hence, using Proposition~\ref{Prep:WildeEuclid} we
  derive the mentioned parameters of the QUENTA code

\end{proof}

\begin{corollary}
    Let $P_1, \ldots, P_n$ be pairwise distinct rational places of
    $F/\mathbb{F}_q$ and $D = P_1 + \cdots + P_n$. Choose
    divisors $G_1, G_2$ of $F/\mathbb{F}_q$ such that ${\supp} (G_i)\cap \supp (D) =
    \varnothing$ and $2g-2<\operatorname{deg}(G_i)<n$, for $i=1,2$.
    If $\operatorname{deg}(G_1^\perp\cup G_2)<n$ and
    $\operatorname{deg}(G_1^\perp\cap G_2) < 0$,
    then there exists a QUENTA code with parameters
    $[[n,\operatorname{deg}(G_2)-g+1, d;c]]_q$, where
    $d\geq n-\max\{\operatorname{deg}(G_1), \operatorname{deg}(G_2)\}$ and
    $c = n+g-1-\operatorname{deg}(G_1)$. In particular, if
    $G_1 = G_2=G$, then the QUENTA code has parameters $[[n,\deg(G)-g+1, d;n+g-1-\deg(G)]]_q$,
    where $d\geq n-\deg(G)$.
\label{newContructionsCorollary:AGQUENTA}
\end{corollary}

The first explicit description of a family of QUENTA codes constructed in this paper is shown in the following theorem. The rational
function field $\mathbb{F}_q(z)/\mathbb{F}_q$ is used to derive this family.

\begin{theorem}
    Let $q$ be a power of a prime. Consider $a_1,a_2, b_1,b_2$ positive integers
    such that $b_1\leq a_2$ and $b_2 \leq q-2-a_2$, with $a_1 + a_2 < q-1$ and $b_1 + b_2 < q-1$,
    then we have the following:

    \begin{itemize}
	  \item If $b_2 \geq a_1 + 1$, then there exists a QUENTA code with parameters\\
	  $$[[q-1, a_1 + b_1-1, \geq q - 1 - \max\{a_1+a_2,b_1 + b_2\}; q-2-(a_2+b_2)]]_{q}.$$
	  \item If $b_2 < a_1 + 1$, then there exists a QUENTA code with parameters\\
	  $$[[q-1, b_1 + b_2+1, \geq q - 1 - \max\{a_1+a_2,b_1 + b_2\}; q-2-(a_1+a_2)]]_{q}.$$
    \end{itemize}
    \label{quentaCodes:rational}
\end{theorem}

\begin{proof}
      Let $\mathbb{F}_q(z)/\mathbb{F}_q$ be the rational function field. The Weil differential
      $\eta = \frac{1}{x^q - x}dx$ satisfies the requirements of Proposition~\ref{AgCodes:DualCodeWeilDiff} and
      it has divisor given by $(\eta) = (q-2)P_\infty - P_0 - D$, where $P_\infty$ and $P_0$ are the place at infinity
      and the origin, respectively, and $D = \sum_{i = 1}^{q-1} P_i$, with $P_i$ being the remaining rational places.
      Assume that $G_1 = a_1 P_0 + a_2 P_{\infty}$ and $G_2 = b_1 P_0 + b_2 P_{\infty}$ and
      $C_1 = C_{\mathcal{L}}(D,G_1)$ and $C_2 = C_{\mathcal{L}}(D,G_2)$.
      Since $b_2 \leq q-2-a_2$, we have that
      $\operatorname{deg}(G_1^\perp\cup G_2)=b_1 + q - 2- a_2 < q-1$, by the hypothesis
      $b_1\leq a_2$, and $G_1^\perp\cap G_2 = (-1-a_1)P_0 + b_2 P_\infty$. Thus, we can use Theorem~\ref{newContructions:AGQUENTA}.
      For the first case, we have that $c = q - 1 - 1 - (a_1+a_2) - (b_2 - a_1) = q-2-(a_2+ b_2)$,
      since $\deg(G_1^\perp\cap G_2)\geq 0$. For the second case, when $b_2 < a_1 + 1$,
      we have that $\deg(G_1^\perp\cap G_2)< 0$, which implies $\ell(G_1^\perp\cap G_2)=0$ and
      $c = q-2-(a_1+a_2)$.
      The remaining claims are derived from Theorem~\ref{newContructions:AGQUENTA} and
      from the observation that $\operatorname{deg}(G_1) = a_1 + a_2$ and
      $\operatorname{deg}(G_2) = b_1 + b_2$.
\end{proof}

\begin{corollary}
      If $b_1 \leq a_2$, $b_2 \leq \min\{q-2-a_2,a_1\}$, with
      $a_1 + a_2=b_1+b_2 < q-1$, there are maximal entanglement almost MDS
      $[[q-1, b_1+b_2+1, \geq q - 1 - (a_1+a_2); q-2-(a_1+a_2)]]_{q}$ QUENTA codes.
      In particular, if $a_1\leq q-2-a_2$, then there exists a
      maximal entanglement MDS
      $[[q-1, a_1+a_2+1, \geq q - 1 - (a_1+a_2); q-2-(a_1+a_2)]]_{q}$ QUENTA codes.
\end{corollary}

\begin{proof}
      Consider the second case of Theorem~\ref{quentaCodes:rational}. Then considering
      $a_1 + a_2=b_1+b_2 < e-2$, the result follows.
\end{proof}

%

The following theorem shows a construction of QUENTA codes derived from the Hermitian function field. Next,
the elliptic function field will be used to obtain maximal entanglement QUENTA codes with Singleton defect at
most one.

\begin{theorem}
    Let $q$ be a power of a prime and $a_1,a_2, b_1,b_2$ be positive integers
    such that $b_1 \leq a_2-q(q-1)$, $b_2 \leq q^3 + q(q-1)-2-a_2$, with $b_1 + b_2 < q^3 -1$ and $a_1 + a_2  < q^3 -1$.
    Then we have the following:

    \begin{itemize}
	  \item If $b_2\geq a_1 + 1$, then there exists a QUENTA code with parameters
	  $$[[q^3-1, a_1 + b_1 +1, \geq q^3 -1 - \max\{a_1 + a_2, b_1 + b_2\}; q^3 - 2 + q(q-1)-(a_2 + b_2) ]]_{q^2}.$$
	  \item If $b_2< a_1 + 1$, then there exists a QUENTA code with parameters
	  $$[[q^3-1, b_1 + b_2 +1-\frac{q(q-1)}{2}, \geq q^3 -1 - \max\{a_1 + a_2, b_1 + b_2\}; q^3 - 2 + \frac{q(q-1)}{2}-(a_1 + a_2) ]]_{q^2}.$$
    \end{itemize}
    \label{quentaCodes:hermitian}
\end{theorem}

\begin{proof}
    Let $F/\mathbb{F}_{q^2}$ be the Hermitian function field defined by the equation
    \begin{equation*}
	y^q + y = x^{q+1}.
    \end{equation*}Then $F/\mathbb{F}_{q^2}$ has $1+q^3$ rational points and genus $g=q(q-1)/2$. Assume that
    $D = P_1 + \cdots + P_{q^3 - 1}$, $G_1 = a_1 P_{0} + a_2 P_\infty$, and
    $G_2 = b_1 P_{0} + b_2 P_\infty$, where $P_\infty$ and $P_{0}$
    are the rational places at infinity and the origin, respectively. Thus, one possible Weil
    differential satisfying Proposition~\ref{AgCodes:DualCodeWeilDiff} is given by $\eta = \frac{1}{x^{q^2}-x}dx$,
    which has divisor
    $(\eta) = -D - P_{0} + (q^3 + q(q-1) - 2)P_\infty$. The fact that
    $b_2 \leq q^3 + q(q-1)-2-a_2$ implies $G_1^\perp\cup G_2 = b_1P_{0} + (q^3 + q(q-1)-2-a_2)P_\infty$.
    By the hypothesis $b_1 \leq a_2-q(q-1)$, we have that $\deg(G_1^\perp\cup G_2)<q^3 - 1$, thus we
    can use Theorem~\ref{newContructions:AGQUENTA}. From this theorem, we derive
    that
    $G_1^\perp\cap G_2 = (-1-a_1)P_{0}+b_2P_\infty$.
    Hence, if $b_2\geq a_1 + 1$ we have that $\deg (G_1^\perp\cap G_2)\geq 0$,
    which implies
    $c = q^3 -1 + \frac{q(q-1)}{2}-1-(a_1+a_2)-(b_2-a_1)+\frac{q(q-1)}{2} = q^3 -2 + q(q-1)-(a_2+b_2)$.
    On the other hand, if $b_2< a_1 + 1$ we have that $\deg (G_1^\perp\cap G_2)< 0$,
    which implies $\ell(G_1^\perp\cap G_2)= 0$ and
    $c = q^3 - 2 + \frac{q(q-1)}{2}-(a_1 + a_2)$.
    Since $\deg(G_1)=a_1+a_2$ and $\deg(G_1)=b_1+b_2$, using Theorem~\ref{newContructions:AGQUENTA}
    and the values of $c$ computed, we derive the mentioned parameters for the QUENTA codes.
\end{proof}


\begin{theorem}
    Let $q = 2^m$, with $m\geq 1 $ an integer. Let $F/\mathbb{F}_q$
    be the elliptic function field with $e$ rational places and genus
    $g = 1$ defined by the equation

    \begin{equation}
	  y^2 + y = x^3 + bx + c,
    \end{equation}where $b,c\in\mathbb{F}_q$.
    Let $a_1,a_2, b_1,b_2$ be positive integers
    such that $b_1 \leq a_2$, $b_2 \leq e-1-a_2$, with $a_1 + a_2  < e-2$ and
    $b_1 + b_2  < e-2$.
    Then we have the following:
    \begin{itemize}
	  \item If $b_2\geq a_1+1$, then there exists a QUENTA code with parameters
	  $$[[e-2, a_1 + b_1 + 1, \geq e-2-\max\{a_1+a_2,b_1+b_2\}; e-1-(a_2+b_2)]]_q.$$
	  \item If $b_2< a_1+1$, then there exists a QUENTA code with parameters
	  $$[[e-2, b_1 + b_2, \geq e-2-\max\{a_1+a_2,b_1+b_2\}; e-2-(a_1+a_2)]]_q.$$
    \end{itemize}
    \label{quentaCodes:elliptic}
\end{theorem}

\begin{proof}
    First of all, let $S = \{\alpha\in\mathbb{F}_q | \text{there exists }\beta\in\mathbb{F}_q
    \text{ such that } \beta^2 + \beta = \alpha^3 + b\alpha + c\}$. For each $\alpha\in S$, there are two
    $\beta\in\mathbb{F}_q$ satisfying the equation $\beta^2 + \beta = \alpha^3 + b\alpha + c$. Thus, for
    each $\alpha\in S$, there are two places corresponding to $x-$coordinate equal to $\alpha$.
    Hence, the set of all rational places is given
    by these $x$ and $y$ coordinates and the place at infinity, $P_{\infty}$.
    The number of rational places is denoted by $e$. So $e = |S| + 1$.
    Now, assume that
    $D = \sum_{i = 1}^{e-2} P_i$,
    $G_1 = a_1P_{0} + a_2 P_{\infty}$, and
    $G_2 = b_1P_{0} + b_2 P_{\infty}$, where
    $P_0, P_{1}, \ldots, P_{e-1}$ are
    pairwise distinct rational places. Additionally, let
    $\eta = \frac{dx}{\prod_{\alpha_i\in S}(x+\alpha_i)}$, then the divisor of the Weil differential $\eta$ is
    given by $(\eta) = (e-1) P_{\infty} - P_{0} - D$.
    The fact that
    $b_2 \leq e-1-a_2$ implies $G_1^\perp\cup G_2 = b_1P_{0} + (e-1-a_2)P_\infty$.
    By the hypothesis $b_1 \leq a_2$, we have that $\deg(G_1^\perp\cup G_2) < e - 2$,
    thus we can use Theorem~\ref{newContructions:AGQUENTA}. From this theorem,
    we derive
    that $G_1^\perp\cap G_2 = (-1-a_1)P_{0}+b_2P_\infty$.
    Hence, if $b_2\geq a_1 + 1$ we have that $\deg (G_1^\perp\cap G_2)\geq 0$,
    which implies $c = e-1-(a_2+b_2)$.
    On the other hand, if $b_2< a_1 + 1$ we have that $\deg (G_1^\perp\cap G_2)< 0$,
    which implies $\ell(G_1^\perp\cap G_2)= 0$ and
    $c = e-2-(a_1+a_2)$.
    Since $\deg(G_1)=a_1+a_2$ and $\deg(G_1)=b_1+b_2$,
    using Theorem~\ref{newContructions:AGQUENTA} and the values of $c$ computed,
    we derive the mentioned parameters for the QUENTA codes.
\end{proof}

\begin{corollary}
      Suppose that there exists an elliptic curve with $e$ rational places. Then for
      $b_1 \leq a_2$, $b_2 \leq \min\{e-1-a_2,a_1\}$, with
      $a_1 + a_2=b_1+b_2 < e-2$, there are maximal entanglement almost MDS
      $[[e-2, b_1+b_2, \geq e - 2 - (a_1+a_2); e-2-(a_1+a_2)]]_{q}$ QUENTA codes.
      In particular, if $a_1+a_2 < e-2$, then there exists a
      maximal entanglement almost MDS
      $[[e-2, a_1+a_2, \geq e - 2 - (a_1+a_2); e-2-(a_1+a_2)]]_{q}$ QUENTA code.
\end{corollary}

\begin{proof}
      Consider the second case of Theorem~\ref{quentaCodes:elliptic}. Then considering
      $a_1 + a_2=b_1+b_2 < e-2$, the result follows.
\end{proof}

It is shown in Table~\ref{table1} the numbers of rational points of several elliptic curves are given
depending on the value of $s$, the degree of the extension $\mathbb{F}_{2^s}$ \cite{Menezes:Book}.

\begin{table}[h]
\begin{center}
\caption{Some elliptic Curves over $\mathbb{F}_q$ ($q = 2^s$)
and their number of rational places \label{table1}}
\begin{tabular}{c|c|c}
\hline\hline
    Elliptic curve                                                                & $s$     & Number of rational places ($e$)\\
    \hline\hline \multirow{3}{5em}{$y^2 + y = x^3$}                               & odd $s$ & $q+1$\\
                                                                                  & $s\equiv 0 \mod 4$ & $q+1 - 2\sqrt{q}$\\
	                                                                          & $s\equiv 0 \mod 2$ & $q+1 + 2\sqrt{q}$\\
    \hline  \multirow{2}{6em}{$y^2 + y = x^3 + x$}                                & $s\equiv 1,7 \mod 8$ & $q+1 + \sqrt{2q}$\\
                                                                                  & $s\equiv 3,5 \mod 8$ & $q+1 - \sqrt{2q}$\\
    \hline  \multirow{2}{8em}{$y^2 + y = x^3 + x + 1$}                            & $s\equiv 1,7 \mod 8$ & $q+1 - \sqrt{2q}$\\
                                                                                  & $s\equiv 3,5 \mod 8$ & $q+1 + \sqrt{2q}$\\
    \hline  $y^2 + y = x^3 + \delta x$ ($Tr (\delta) = 1$)                        & $s$ even & $q+1$\\
    \hline  \multirow{2}{13em}{$y^2 + y = x^3 + \delta$ ($Tr (\delta) = 1$)}      & $s\equiv 0 \mod 4$ & $q+1 + 2\sqrt{q}$\\
                                                                                  & $s\equiv 2 \mod 4$ & $q+1 - 2\sqrt{q}$\\
\hline
\end{tabular}
\end{center}
\end{table}


\begin{remark}
In this section, two-point AG codes have been used to construct QUENTA codes. The reason for this is that
the QUENTA codes derived from one-point AG codes have trivial parameters; i.e., they have either zero entanglement,
which it are codes that are derived from the standard quantum stabilizer construction
(e.g. quantum codes from the CSS construction \cite{NielsenChuang:Book}),
or zero dimension, what make them not interesting for this paper.
\end{remark}

\subsection{Hermitian Construction}


In opposition to the Euclidean dual of an AG code, there is no general formula to describe the Hermitian dual of an AG code
as in Definition~\ref{AgCodes:definitionAG}. However, describing an AG code via a basis of evaluated elements that belong to
a Riemann-Roch space, we can obtain the information that we need from the Hermitian dual code. Before doing that,
our approach to calculate the dimension of the intersection between an AG code and its Hermitian dual is shown.

\begin{proposition}
        Let $C$ be a linear code over $\mathbb{F}_{q^2}$ with length $n$ and $C^{\perp_h}$ its dual. Then
        $\dim(C\cap C^{\perp_h}) = \dim(C^\perp\cap C^q)$.
        \label{newConstructions:IntersectionCodes}
\end{proposition}

 \begin{proof}
        Although it is well known that $C^{\perp_h} = (C^\perp)^q$ \cite{Huffman:2003}, we present this result here for completeness

        \begin{eqnarray*}
                \bx\in C^{\perp_h} &\iff \bx \cdot \bc^q = 0, \quad &\forall \bc \in C,\\
                			       &\hspace{0.80cm}\iff \sum_{i=1}^n x_i c_i^q = 0, \quad &\forall \bc \in C,\\
                			       &\hspace{0.80cm}\iff \sum_{i=1}^n x_i^q c_i = 0, \quad &\forall \bc \in C,\text{ since }c_i^{q^2} = c_i\\
                			       &\hspace{-0.50cm}\iff \bx^q \in C^\perp,\\
                			       &\hspace{-0.22cm}\iff \bx \in (C^\perp)^q.
        \end{eqnarray*}Thus, we see that

        \begin{eqnarray*}
                \dim(C \cap C^{\perp_h})&=& \dim(C \cap (C^\perp)^q)\\
                	                    &=& \dim(C^q \cap C^\perp).
        \end{eqnarray*}Hence, we have $\dim(C\cap C^{\perp_h}) = \dim(C^\perp\cap C^q)$.
\end{proof}

Proposition~\ref{newConstructions:IntersectionCodes} shows a new way to compute the dimension of $C \cap C^{\perp_h}$.
To be able to use it for AG codes, we need to describe the linear code $C^q$.
Proposition~\ref{newConstructions:basisforCq}
approaches this by showing that it is possible to compute a basis to $C^q$ from a basis of $C$.
And Theorem~\ref{newConstructions:DimHermitianIntersection} describes how to compute the intersection
of two vector space (in particular, linear codes) when the basis of each one belongs to the
same larger set of which is a basis of $\mathbb{F}_{q^2}^n$.

\begin{proposition}
        Let $C$ be a linear code over $\mathbb{F}_{q^2}$ with length $n$ and dimension $k$. If $\{\bx_1, \ldots, \bx_k\}$
        is a basis of $C$, then a basis of $C^q$ is given by the set $\{\bx_1^q, \ldots, \bx_k^q\}$.
        \label{newConstructions:basisforCq}
\end{proposition}

\begin{proof}
        First of all, notice that for any $\bc'\in C^q$ there is a unique $\bc\in C$ such that $\bc' = \bx^q$. Since that
        $\{\bx_1, \ldots, \bx_k\}$ is a basis for $C$, then $\bc = \sum_{i=1}^{k}c_i \bx_i$, for $c_i\in\mathbb{F}_{q^2}$.
        Therefore, $\bc' = \bc^q = \sum_{i=1}^{k}c_i^q \bx_i^q = \sum_{i=1}^{k}c'_i \bx_i^q$. Since that $C$ and $C^q$ are isomorphic,
        then they have the same dimension which implies that $\{\bx_1^q, \ldots, \bx_k^q\}$ is a basis for $C^q$.
\end{proof}


For sake of self-contained, we present the following lemma.

\begin{lemma}
        Let $B$ be a basis for $\mathbb{F}_{q^2}^n$ and $B_1$ and $B_2$ be two subsets of $B$.
        Denoting by $V_1$ and
        $V_2$ the subspaces generated by $B_1$ and $B_2$, respectively, then we have that
        $\dim(V_1\cap V_2) = |B_1\cap B_2|$.
        \label{newConstructions:DimHermitianIntersection}
\end{lemma}

\begin{proof}
        The claim that any element in $B_1\cap B_2$ gives a vector in
        $V_1\cap V_2$ is trivial. Denote the elements of $B_1, B_2$, and $B_1\cap B_2$ by ${\bf v}_{1i},{\bf v}_{2i}$ and ${\bf v}_{0i}$, respectively. In order to prove the
        reverse inclusion we consider $\bv\in V_1\cap V_2$.
        Thus, we can represent it as

        \begin{equation*}
                \bv = \sum_{{\bf v}_{0i}\in B_1\cap B_2}c_{0i}{\bf v}_{0i} + \sum_{{\bf v}_{1i}\in B_1\setminus (B_1\cap B_2)}c_{1i}{\bf v}_{1i}\qquad \text{and}\qquad
                \bv = \sum_{{\bf v}_{0i}\in B_1\cap B_2}d_{0i}{\bf v}_{0i} + \sum_{{\bf v}_{2i}\in B_2\setminus (B_1\cap B_2)}d_{2i}{\bf v}_{2i},
        \end{equation*}
        which implies that
        \begin{equation*}
                \sum_{{\bf v}_{0i}\in B_1\cap B_2}(c_{0i} - d_{0i}){\bf v}_{0i} + \sum_{{\bf v}_{1i}\in B_1\setminus (B_1\cap B_2)}c_{1i}{\bf v}_{1i} - \sum_{{\bf v}_{2i}\in B_2\setminus (B_1\cap B_2)}d_{2i}{\bf v}_{2i}=0.
        \end{equation*}
        Since ${\bf v}_{0i}, {\bf v}_{1i}$, and ${\bf v}_{2i}$ belong to the basis $B$, we have that every coefficient in the
        previous equation needs to be equal to zero, which results in $\bv = \sum_{{\bf v}_{0i}\in B_1\cap B_2}c_{0i}{\bf v}_{0i}$
        and the reverse inclusion is proved.
\end{proof}

Now, we derive QUENTA codes from the Hermitian construction using AG codes. To illustrate this, we are going
to apply the result from Lemma~\ref{newConstructions:DimHermitianIntersection} to AG codes derived from
rational function field. See Theorem~\ref{quentaCodes:rationalDualRational}.

\begin{theorem}
      Let $q$ be a prime power and $m$ an integer which is written as
      $m = qt + r < q^2$, where $t\geq 1$ and $0 \leq r \leq q-1$.
      Then we have the following:
      \begin{itemize}
	  \item If $t\geq q-r-1$, then there exists an MDS QUENTA code with parameters
	  $$[[q^2, (t+1)^2 + 2r + 1 - 2q, \geq q^2-(qt+r); (q-t-1)^2]]_q.$$
	  \item If $t< q-r-1$, then there exists an MDS QUENTA code with parameters
	  $$[[q^2, t^2 - 1 , \geq q^2-(qt+r); (q-t)^2 - 2(r+1)]]_q.$$
      \end{itemize}

      \label{quentaCodes:rationalDualRational}
\end{theorem}

\begin{proof}
      Let $F(z)/\mathbb{F}_{q^2}$ be the rational function field, $D = \sum_{i=0}^{q^2-1} P_i$ and
      $G = mP_\infty$, where $m = qt+r$. Let $C_{\mathcal{L}}(D,G)$ be the AG code derived from $D$ and $G$ with
      parameters $[q^2, m+1, q^2 - m]_q$.
      Consider $\vec{x^i} = ev_D(x^i)$. Let $B = \{\vec{x^i}|0\leq i \leq n-1\}$. Then $B$ is a basis of
      $\mathbb{F}_{q^2}^n$. A basis for $C_{\mathcal{L}}(D,G)$ is given by a subset,
      $B' = \{\vec{x^i}|0\leq i \leq m\}$. Thus, a basis of $C_{\mathcal{L}}(D,G)^q$ can be described as
      $B_1 = \{\vec{x^{qi}}|0\leq i \leq m\}$. Now, notice that $\vec{x^{q^2 + a}} = \vec{x^{a+1}}$ for
      all $a\geq 0$. Therefore, $$B_1 = \{\vec{x^{qi+j}}|0\leq i \leq q-1, 0\leq j\leq t-1\}\cup \{\vec{x^{qi+t}}|0\leq i \leq r\}.$$
      On the other hand, a basis of $C_{\mathcal{L}}(D,G)^\perp$ is given by the set
      $$B_2 = \{\vec{x^i}|0\leq i \leq q^2 - 2 - m\} = \{\vec{x^{qi+j}}|0\leq i \leq q - t - 2, 0\leq j \leq q-1\}\cup
      \{\vec{x^{(q-t-1)q + j}}| 0\leq j \leq q-r-2\}.$$ Thus, the exponents of $\vec{x}$ in the bases $B_1$ and $B_2$ can
      be represented by the sets

      \begin{equation*}
	    \begingroup 
	    \setlength\arraycolsep{3pt}
            \begin{Bmatrix}
	    0        & 1        & 2        & \cdots & t-1        & t       \\
	    q        & q+1      & q+2      & \cdots & q+t-1      & q+t     \\
	    \vdots   & \vdots   & \vdots   & \cdots & \vdots     & \vdots \\
	    rq       & rq+1     & rq+2     & \cdots & rq+t-1     & rq + t  \\
	    \vdots   & \vdots   & \vdots   & \cdots & \vdots     & \\
	    (q-1)q   & (q-1)q+1 & (q-1)q+2 & \cdots & (q-1)q+t-1 & \\
	    \end{Bmatrix}.
	    \endgroup
      \end{equation*}and

      \begin{equation*}
	    \begingroup 
	    \setlength\arraycolsep{3pt}
            \begin{Bmatrix}
	    0        & 1          & \cdots & q-r-2          & \cdots &q-1\\
	    q        & q+1        & \cdots & q+q-r-2        & \cdots & 2q-1    \\
	    \vdots   & \vdots     & \cdots & \vdots         & \vdots & \vdots\\
	    (q-t-2)q & (q-t-2)q+1 & \cdots & (q-t-2)q+q-r-2 & \cdots & (q-t-2)q+q-1\\
	    (q-t-1)q & (q-t-1)q+1 & \cdots & (q-t-1)q+q-r-2 &        &
	    \end{Bmatrix},
	    \endgroup
      \end{equation*}respectively. Using this description, we see that these bases satisfy the hypothesis in Lemma~\ref{newConstructions:DimHermitianIntersection}, so it is possible to compute the intersection of the codes related to $B_1$ and $B_2$ via the computation of the intersection of these sets. To do so, we have to consider
      two cases separately, $t\geq q-r-1$ and $t< q-r-1$. For the first case, the intersection is given by the following set
      $$B_1\cap B_2 = \{\vec{x^{qi+j}}| 0\leq i \leq q-t-2, 0\leq j \leq t\}\cup \{\vec{x^{(q-t-1)q+j}}| 0\leq j \leq q-r-2\}.$$ Thus,
      $\dim (C_{\mathcal{L}}(D,G)^q\cap C_{\mathcal{L}}(D,G)^\perp) = |B_1\cap B_2| = (q-t-1)(t+1) + q-r-1$. Using the same description for the case $t< q-r-1$, we see that
      $$B_1\cap B_2 = \{\vec{x^{qi+j}}| 0\leq i \leq q-t-1, 0\leq j \leq t-1\}\cup \{\vec{x^{(qi+t}}| 0\leq i \leq r\},$$ which implies
      $\dim (C_{\mathcal{L}}(D,G)^q\cap C_{\mathcal{L}}(D,G)^\perp) = |B_1\cap B_2| = (q-t)t + r+1$. Applying the previous computations, and using the fact that
      $C_{\mathcal{L}}(D,G)$ has parameters $[q^2, m+1, q^2 - m]_q$, by Proposition~\ref{Prep:WildeHerm}, we have that there exists a QUENTA code with parameters

      \begin{itemize}
          \item $[[q^2, (t+1)^2 + 2r + 1 - 2q, \geq q^2-(qt+r); (q-t-1)^2]]_q$, for $t\geq q-r-1$; and
          \item $[[q^2, t^2 - 1 , \geq q^2-(qt+r); (q-t)^2 - 2(r+1)]]_q$, for $t< q-r-1$.
      \end{itemize}
      \end{proof}

\section{Code Comparison}
\label{sec:codComp}
In Tables~\ref{table2}~and~\ref{table3}, we present some optimal QUENTA codes obtained
from Theorems~\ref{quentaCodes:rational},~\ref{quentaCodes:elliptic}~and~\ref{quentaCodes:rationalDualRational}.
The QUENTA codes derived from the Euclidean construction are presented in Table~\ref{table2}. We use
AG codes obtained from projective line and elliptic curves to construct these codes. As can be seen,
the codes in the first column of Table~\ref{table2} are MDS, and the ones in the second are almost MDS. For
Table~\ref{table3}, the QUENTA codes are derived from the Hermitian construction, where rational AG codes
were used as the classical code. These codes have also an optimal combination of parameters,
since they are MDS. Additionally, since QUENTA codes use entanglement,
we conclude that these quantum codes from
Tables~\ref{table2}~and~\ref{table3} have better or equal minimal distances than any quantum code
with the same length and dimension derived from quantum stabilizer codes \cite{NielsenChuang:Book}.

\begin{table}[h]
\begin{center}
\caption{Some new maximal entanglement (almost) MDS QUENTA codes from the Euclidean construction\label{table2}}
\begin{tabular}{c|c}
\hline\hline
New QUENTA codes -- Theorem~\ref{quentaCodes:rational}                 & New QUENTA codes -- Theorem~\ref{quentaCodes:elliptic} \\
$[[q-1, a_1 + a_2 + 1, q - 1 - (a_1 + a_2); q - 2 - (a_1+ a_2)]]_{q}$  & $[[e-2, a_1 + a_2, e - 2 - (a_1 + a_2); e - 2 -  (a_1 + a_2)]]_q$ \\
$a_1 +a_2\leq q-2 $                                                    & $a_1 + a_2  < e-2$, and $e$ as in Table~\ref{table1}\\
\hline \multicolumn{2}{c}{Examples}\\ \hline
\hline ${[[3, 2, 2; 1]]}_{4}$                                 & ${[[7, 5, 2; 2]]}_{4}$ \\
\hline ${[[4, 2, 3; 2]]}_{5}$                                 & ${[[11, 6, 5; 5]]}_{8}$ \\
\hline ${[[6, 4, 3; 2]]}_{7}$                                 & ${[[11, 8, 3; 3]]}_{8}$ \\
\hline ${[[7, 4, 4; 3]]}_{8}$                                 & ${[[23, 13, 10; 10]]}_{16}$ \\
\hline ${[[10, 7, 4; 3]]}_{11}$                               & ${[[23, 18, 5; 5]]}_{16}$ \\
\hline ${[[12, 7, 6; 5]]}_{13}$                               & ${[[39, 25, 14; 14]]}_{32}$ \\
\hline ${[[15, 10, 6; 5]]}_{16}$                              & ${[[39, 18, 11; 11]]}_{32}$ \\
\hline
\end{tabular}
\end{center}
\end{table}

\begin{table}[h]
\begin{center}
\caption{Some new MDS QUENTA codes from Hermitian construction\label{table3}}
\begin{tabular}{c}
\hline\hline
New QUENTA codes -- Theorem~\ref{quentaCodes:rationalDualRational}\\
$[[q^2, (t+1)^2 + 2r + 1 - 2q, q^2-(qt+r); (q-t-1)^2]]_q$\\
$m = qt + r < q^2$, $t\geq q-r-1$ and $0 \leq r \leq q-1$ \\
\hline Examples\\ \hline
\hline ${[[16, 6, 6; 1]]}_{4}$\\
\hline ${[[49, 25, 13; 1]]}_{7}$\\
\hline ${[[49, 11, 24; 9]]}_{7}$\\
\hline ${[[64, 29, 20; 4]]}_{8}$\\
\hline ${[[64, 25, 22; 4]]}_{8}$\\
\hline ${[[81, 33, 29; 9]]}_{9}$\\
\hline ${[[81, 16, 41; 16]]}_{9}$\\
\hline ${[[256, 141, 66; 16]]}_{16}$\\
\hline
\end{tabular}
\end{center}
\end{table}

The remaining QUENTA codes that are compared with the literature are the ones
derived from Hermitian curve. The first analysis of goodness of our codes is via
the Singleton defect, which is the difference between the quantum
Singleton bound (QSB) presented in Eq.~\ref{QSB} and the minimal
distance of the code.
Recall that an $[[n, k, d; c]]_{q}$ quantum code
satisfies $k + 2d \leq n + c + 2$ (QSB). Hence, the codes derived from
Theorem~\ref{quentaCodes:hermitian}
have maximum Singleton defect equals to $q(q-1)/2$.
Some examples of parameters derived are
$[[7,2,4;3]]_{4}$, $[[26,10,11;10]]_{9}$, and $[[63,19,32;31]]_{16}$
which have Singleton defect $1$, $3$, and $6$, respectively. Comparing
these examples with quantum stabilizer codes, we see that our codes have
minimal distance unreachable for the same length and dimension. This can
be seen from the quantum Singleton bound for stabilizer codes (a more
general case of quantum codes). Thus, even though the codes from Theorem~\ref{quentaCodes:hermitian}
are not MDS with respect to its quantum Singleton bound, they can be used to attain
parameters that are unreachable by quantum stabilizer codes. Adopting entanglement defect
as been equal to the difference between the actual amount of entanglement in the QUENTA code
and $n-k$, we see that the entanglement defect in this family of QUENTA code is
equal to $2g$, where $g$ is the genus of the Hermitian function field. Lastly, Table~\ref{table4} shows
some examples of QUENTA codes that have a higher rate than the asymptotic Gilbert-Varshamov bound presented
in Section~\ref{sec:AsymptoticallyGood}.

\begin{table}[h]
\begin{center}
\caption{Some new QUENTA codes from Hermitian Curve\label{table4}}
\begin{tabular}{c}
\hline\hline
New QUENTA codes -- Theorem~\ref{quentaCodes:hermitian}                        \\
$[[q^3-1, b_1 + b_2 +1 - q(q-1)/2, q^3 -1 -\max\{a_1 + a_2, b_1 + b_2\};$   \\
$q^3 + q(q-1)/2 - (a_1 + a_2) - 2]]_{q^2}$                                  \\
$b_1\leq a_2 - q(q-1)$, $b_2\leq\min\{q^3 + q(q-1)-2-a_2, a_1\}$, and $a_1+a_2, b_1+b_2 < q^3-1$ \\
\hline Examples\\ \hline
\hline ${[[26, 15, 6; 5]]}_{9}$                                                  \\
\hline ${[[64, 39, 11; 10]]}_{16}$                                                  \\
\hline ${[[125, 51, 54; 53]]}_{25}$                                                  \\
\hline ${[[343, 179, 122; 121]]}_{49}$                                                  \\
\hline
\end{tabular}
\end{center}
\end{table}

\newpage
\section{Asymptotically Good Maximal Entanglement QUENTA Codes}
\label{sec:AsymptoticallyGood}
In this section, we show that from any family of (classical) asymptotically good AG codes,
we can construct a family of asymptotically good maximal entanglement QUENTA codes.
This is a consequence of the use of the result from Carlet, \emph{et al.}
\cite{Carlet_Pellikaan:2018} applied to the Corollary~\ref{corollary2}. Before showing it, we need to define
the concept of (classical) asymptotically good codes.

\begin{definition}
      Let $q$ be a prime power and $\alpha_q := \sup\{R\in [0,1]\colon (\delta, R)\in U_q\}$, for
      $0\leq \delta \leq 1$. Here $U_q$ denotes the set of all ordered pair $(\delta, R)\in[0,1]^2$ for which
      there is a family of linear codes that are indexed as $C_i$, with parameters $[n_i,k_i,d_i]_q$, such that
      $n_t\rightarrow \infty$ as $t\rightarrow\infty$ and $\delta = \lim_{i\rightarrow \infty} d_i/n_i$,
      $R = \lim_{i\rightarrow \infty} k_i/n_i$. If $\delta, R > 0$, then the family is called asymptotically good.

\end{definition}

\begin{proposition}\cite[Corollary 14]{Carlet_Pellikaan:2018}
      Let $q> 3$ be a power of a prime and $A(q) = \lim\sup_{g\rightarrow \infty}\frac{N_q(g)}{g}$, where
      $N_q(g)$ denotes the maximum number of rational places that a global function field of genus $g$ with
      full constant field $\mathbb{F}_q$ can have. Then there exists a family of LCD codes with

      \begin{equation}
	    \alpha_q^{LCD} (\delta) \geq 1-\delta - \frac{1}{A(q)}, \text{ for }\delta\in[0,1].
      \end{equation}

%
%
%
%

\label{Proposition:Carlet_Pellikaan}
\end{proposition}

\begin{theorem}
      Let $q> 3$ be a power of a prime and $A(q)$ as defined in Proposition~\ref{Proposition:Carlet_Pellikaan}.
      Then there exists a family of asymptotically good maximal entanglement QUENTA codes with parameters $[[n_t,k_t,d_t;c_t]]_q$, such that
      \begin{equation*}
	    \lim_{t\rightarrow \infty}\frac{d_t}{n_t} \geq \delta, \qquad \lim_{t\rightarrow \infty}\frac{k_t}{n_t} \geq 1 - \delta -\frac{1}{A(q)},
      \end{equation*}and
      \begin{equation*}
        \lim_{t\rightarrow \infty}\frac{c_t}{n_t}\in [\delta, \delta + 1/A(q)].
      \end{equation*}for all $\delta\in[0,1-1/A(q)]$.
\label{asymp_eaqecc}
\end{theorem}

\begin{proof}
Let $\mathcal{C} = \{C_1, C_2, \ldots\}$ be a family of asymptotically good LCD codes as the ones in
Proposition~\ref{Proposition:Carlet_Pellikaan}, where each $C_t$ has parameters
$[n_i, k_i, d_i]_q$. If we apply the family $\mathcal{C}$ to construct QUENTA codes, it follows
from Corollary~\ref{corollary2} that we can construct maximal entanglement QUENTA codes with
parameters $[[n_t,k_t,d_t;c_t]]_q$, such that
\begin{eqnarray*}
	    \lim_{t\rightarrow \infty}\frac{d_t}{n_t} =  \lim_{i\rightarrow \infty}\frac{d_i}{n_i} \geq \delta,\hspace{2cm}
        \lim_{t\rightarrow \infty}\frac{k_t}{n_t} = \lim_{i\rightarrow \infty}\frac{k_i}{n_i} \geq 1 - \delta -\frac{1}{A(q)}.
\end{eqnarray*}Moreover, we have
\begin{equation*}
        \lim_{t\rightarrow \infty}\frac{c_t}{n_t} = \lim_{i\rightarrow \infty}\frac{n_i - k_i}{n_i} = \lim_{i\rightarrow \infty}1 - \frac{k_i}{n_i}\leq \delta +\frac{1}{A(q)}
\end{equation*}and
\begin{equation*}
        \lim_{t\rightarrow \infty}\frac{c_t}{n_t} = \lim_{i\rightarrow \infty}\frac{n_i - k_i}{n_i} \geq \lim_{i\rightarrow \infty}\frac{d_i-1}{n_i}\geq \delta,
\end{equation*}for $\delta\in[0,1-1/A(q)]$. Thus, since that the families in Proposition~\ref{Proposition:Carlet_Pellikaan} are asymptotically good, then the family of
QUENTA codes is asymptotically good maximal entanglement.
\end{proof}

\begin{remark}
If $q$ is a square, then $A(q) = \sqrt{q}-1$ by \cite{Tsfasman:1982,Ihara:1981}.
\end{remark}

\begin{remark}
In a recent paper, Galindo, \emph{et al.} \cite{Galindo:2019} derived the quantum Gilbert-Varschamov bound for
QUENTA codes. Using AG codes derived from tower of function fields that attain the Drinfeld-Vladut bound
\cite{Stichtenoth:2009} and the previous theorem, we can show that there is a family of QUENTA codes
with parameters that exceed the mentioned bound (see Figure~\ref{Fig:GVBound}).
\end{remark}

\begin{figure*}[h!]
\centering
\includegraphics[width=0.7\textwidth]{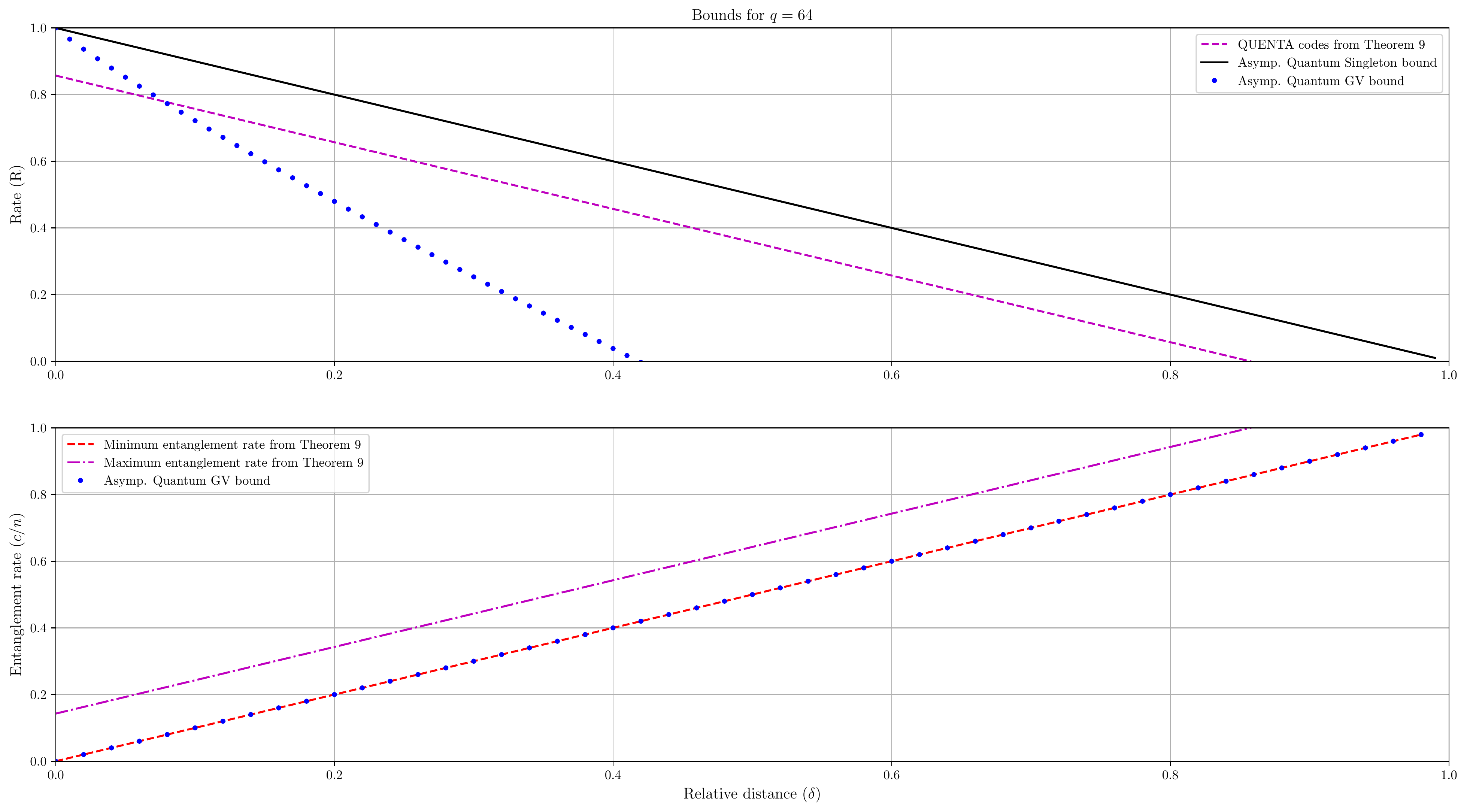}
\caption{Comparison between QUENTA codes derived from Theorem~\ref{asymp_eaqecc} and quantum
Gilbert-Varshamov bound of \cite{Galindo:2019} via analysis of rate and relative entanglement when
$q = 64$.}
\label{Fig:GVBound}
\end{figure*}



\section{Conclusion}
\label{sec:Conclusion}
This paper has been devoted to the use of AG codes in the construction of QUENTA codes.
We firstly showed two methods to create new AG codes from old ones via intersection and
union of divisors. Afterwards, the former method is applied to construct quantum codes via
the Euclidean construction method for QUENTA codes. Two of the families derived in this part are
MDS or almost MDS and, for some particular range of parameters, have maximal entanglement. For the QUENTA codes constructed from
the Hermitian function field, we have shown that it is possible to achieve higher rates when compared with
standard quantum stabilizer codes and the entanglement-assisted quantum Gilbert-Varshamov bound.
In the following, using the Hermitian construction
method for QUENTA codes, we have constructed one more family of QUENTA codes from AG codes, which was also shown to
be MDS. Lastly, it was shown that for any asymptotically good family of
classical codes, there is a family of asymptotically good maximal entanglement QUENTA codes. In addition,
it is demonstrated that there are QUENTA codes surpassing the quantum Gilbert-Varshamov bound.


\section{Acknowledgements}
This work was supported by the \emph{Conselho Nacional de Desenvolvimento
Cient\'ifico e Tecnol\'ogico}, grant No. 201223/2018-0.

%
\bibliographystyle{splncs04}
\bibliography{ref}

\end{document}